\documentclass[12pt]{article}
\usepackage{comment}

\usepackage{setspace}
\usepackage[utf8]{inputenc}
\usepackage[a4paper, tmargin=1.1in, lmargin=0.8in, rmargin=0.8in, bmargin=1.1in]{geometry}
\usepackage{url}
\usepackage{xcolor}
\usepackage{amsmath}
\usepackage{bbm}
\usepackage{amsfonts}
\usepackage{amssymb}
\usepackage{amsthm}
\usepackage[english]{babel}
\usepackage{lscape} 
\usepackage{hyperref}
\usepackage{csquotes}

\usepackage{dsfont} 

\usepackage{natbib}

\usepackage{float}
\usepackage[pdftex]{graphicx}

\newtheorem{proposition}{Proposition}

\usepackage{fancyhdr}
\title{Small Area Estimation of Inequality Measures using Mixtures of Beta}

\author{Silvia De Nicolò\thanks{\texttt{silvia.denicolo@unibo.it}} \and Maria Rosaria Ferrante \and Silvia Pacei}

\date{University of Bologna}

\begin{document}

\maketitle
\thispagestyle{fancy}
\begin{abstract}
Economic inequalities referring to specific regions are crucial in deepening spatial heterogeneity. Income surveys are generally planned to produce reliable estimates at countries or macroregion levels, thus we implement a small area model for a set of inequality measures (Gini, Relative Theil and Atkinson indexes) to obtain reliable microregion estimates. Considering that inequality estimators are unit-interval defined with skewed and heavy-tailed distributions, we propose a Bayesian hierarchical model at the area level involving a Beta mixture. An application on EU-SILC data is carried out and a design-based simulation is performed. Our model outperforms in terms of bias, coverage and error the standard Beta regression model. Moreover, we extend the analysis of inequality estimators by deriving their approximate variance functions.
\end{abstract}

\section{Introduction}

The issue of widespread economic inequality characterizes the current global predicament and has a central role in political and economic discourse \citep{piketty2014capital}. The demand for inequality estimates referring to specific subpopulations is growing, policymakers and stakeholders need them in order to formulate and implement policies, distribute resources and measure the effect of policy actions at local level. Besides, such estimates may be valuable in order to further deepen some research trends in regional and inequality studies, for instance, to identify which regions constitute the drivers of national income inequality and to study spatial spillovers \citep{moser2017income, marquez2019role}. For a recent review on spatial inequality, see \citet{cavanaugh2018locating}.

Economic inequality is conventionally measured on equivalent disposable income data.
Generally, such data are collected via household sample surveys which are planned for aggregates estimation at coarse geographical levels, being scarcely available at finer geographical scales. Thus, local domains fall outside the prior design plan, resulting in small-sized samples and yielding unreliable direct estimation (i.e. with large error).
 For instance, the Survey of Income and Living Conditions (EU-SILC), which provides information on income for the whole set of European countries, is able to provide reliable estimates only at NUTS-2 as the maximum level of disaggregation. The problem could be overcome by increasing the survey sample size, but it is often excluded by cost–benefit analysis. A solution to cope with it is to rely on Small Area Estimation (SAE) techniques. Such techniques exploit auxiliary information to borrow strength across areas and produce area-specific estimates with an acceptable level of uncertainty. The model-based class of SAE techniques leverages hierarchical models, both at area level or unit (individual) level, producing reliable estimates when models are correctly specified and informative auxiliary variables are available without error. For a review, see \citet[ch. 4]{rao2015} and \citet{tzavidis2018start}.

In literature, only a few contributions relate to the small area estimation of inequality indicators, since in this context poverty has special consideration and inequality is treated as a minor appendix \citep{molina2010small, pratesi2016analysis}.
 \citet{Fabrizi2016} deal with the Gini index estimation via area level models. Other authors propose unit level models: \citet{tzavidis2016robust} for Gini index and Quantile Share Ratio, \citet{marchetti2021robust} for Gini and Theil indexes and \citet{gardini2022poverty} for the Quantile Share Ratio. All of them never treat more than two inequality measures at a time.
 However, inequality can be seen as a multifaceted concept, embracing diverse objective and normative assessments on the characteristics of income distribution \citep{cowell2011measuring}.
 It can be measured via a plethora of statistical indicators, all of them with different axiomatic properties and featuring varying sensitivities to extreme values and income transfers that could reduce inequality. Thus, the point of concurrently producing estimates of various indicators could provide a comprehensive overview of the phenomenon.

 In this spirit, we propose a small area estimation strategy for a set of four income inequality measures. In addition to the Gini index, we consider the Relative Theil index, particularly appealing due to its additive decomposability property which allows expressing inequality as the sum of between and within components. Moreover, we consider the family of Atkinson measures, which stands apart from other descriptive measures by explicitly incorporating a welfare evaluation of inequality implications and enabling a complete ordering of income distributions.
 All the measures considered vary between 0 (case of perfect equality) and 1 (perfect inequality), having double bounded support, where degenerate cases 0 and 1 have probability very close to zero also in a small sample context.

 Our approach lies within the framework of Bayesian inference of area level models. Generally, unit level proposals require the auxiliary variables to be defined for the entire population. This may be difficult to achieve as administrative archives are not easily accessible and linked to each other and to survey samples \citep{harmeningframework}. In contrast, area level models require the auxiliary information to be defined only at domain level. For this reason, an area level proposal may be valuable, being less demanding with respect to data requirements as well as incorporating design-based properties in a straightforward way.


Inequality design-based estimators have peculiar characteristics; their behaviour in complex survey small samples has been investigated showing highly skewed and heavy-tailed distributions at decreasing sample sizes, even conditionally on auxiliary variables. A Monte Carlo analysis is available in Section S.2 of the Supplementary Material. In this context, they may require a flexible regression specification. Previous proposals in this sense exploit semi-parametric specifications at the unit level \citep{opsomer2008non, bianchi2018estimation}, suitable data transformations \citep{rojas2020data} or alternative likelihood assumptions. At the area level, the inadequacy of Fay-Herriot models \citep{fay1979estimates} in case of skewed and heavy-tailed estimators is particularly established. The alternative likelihoods considered are the skew-normal \citep{Ferraz2012, ferrante2017small}, skew-t \citep{moura2017small} and log-normal ones \citep{slud2006mean, fabrizi2018bayesian}. When dealing with unit interval-defined responses, such alternative specifications may fit values outside the variable support.
In this case, the small area literature at the area level gathers linear mixed models with suitable transformations and Beta regression models \citep{Janicki2020}.
However, even the Beta distribution may fail in case of markedly skewed and heavy-tailed estimators \citep{bayes2012new, Migliorati2018}. The body of literature based on Beta regression in SAE comprises univariate proposals \citep{liu2007hierarchical, bauder2015small, Fabrizi2016}; multivariate versions  \citep{Souza2016, Fabrizi2011}; and zero and/or one inflated extensions \citep{Wieczorek2012, fabrizi2020functional}.

Our proposal involves incorporating an alternative likelihood assumption by adopting a Beta mixture-based approach, whose performances are compared with the most common Beta regression model. Specifically, we assume as sampling distribution the Flexible Beta one, proposed by \citet{Migliorati2018}. The Flexible Beta distribution is a mixture of two Beta random variables,
particularly interesting for the purpose of small-area estimating inequality measure due to its superior flexibility, given its four parameters structure. Indeed, the Beta distribution has good properties, being able to adapt to different shapes, but its two parameters structure hinders further flexible modelling. Eventually, we derive the approximate variance function of each inequality estimator, analyzing how their mean and variance are tied together and whether such interrelation differs among measures.

Our contribution has therefore multiple levels. On one hand, we provide a comprehensive discussion about inequality and its SAE by considering a set of multiple measures. Secondly, we deepen the analysis of inequality estimators by deriving their approximate variance functions, which may be useful for further modelling.
Thirdly, our methodological proposal extends small area literature in the case of unit interval-defined, skewed and heavy-tailed estimators.
Our model comes out
to outperform the Beta one, both in terms of bias and error of the target estimators, avoiding to highly
underestimate inequality and providing reliable estimates.

The paper is organized as follows. Inequality measures and their estimators are defined and described in Section \ref{inequalitymeasure}, together with a proposal of sampling variance estimation. Section \ref{smallareamodels} defines the proposed Beta and Flexible Beta small area models. An application on EU-SILC income data is unraveled
in Section \ref{application} and a design-based simulation can be found in Section \ref{designbasedsimulation}, in order to evaluate the frequentist properties of model-based estimators. Conclusions are drawn in Section \ref{conclusions}.

\section{Inequality Measures}
\label{inequalitymeasure}

In this section, we describe the inequality measures considered and their estimator in the complex survey case. Such estimators are known to be biased in small samples, often leading to underestimation; therefore, we adopt the bias-corrected direct estimators proposed by \citet{Denicolo2021}. Their simulation results lead us to assume that such estimators are approximately unbiased or slightly biased depending on the domain sample size. In Subsection \ref{variancestimation}, their variance estimation is set out and their estimates are commented.

The most famous inequality measure is the Gini index,
 measuring concentration in the distribution of a positive random variable; among its several equivalent definitions, we adopt the formulation of \citet{sen1997economic}. Suppose we are dealing with a finite population, denoted with $\mathcal{U}$, of $N (< \infty)$ elements, and let a sample $\mathcal{S}_{iid}$ of size $n$ be randomly drawn from $\mathcal{U}$. Let $z \in \mathbb{R}^+$ be a characteristic of interest, in our case income,  which is observable for each unit in $\mathcal{S}_{iid}$.  The simple random sampling (srs) estimator of the Gini index is defined as

\begin{align*}
G=\frac{2}{n^2}\frac{ \sum_{i \in \mathcal{S}_{iid}} z_i (r_i -1/2)}{ \hat{\mu}_{srs}}-1,
\end{align*}
\noindent with $r_i$ the rank of the $i$-th sampled unit and $\hat{\mu}_{srs}$ is the sample mean. Let us suppose, moreover, that a sample $\mathcal{S}$ is drawn from $\mathcal{U}$ through a complex selection scheme e.g., involving stratification and multi-stage selection, as in the case of survey data. This involves unequal inclusion probabilities across units, thus a weighted estimator should be adopted, as proposed by \citet{Langel2013}:
\begin{align*}
G_w&=\frac{2 \sum_{i \in \mathcal{S}} w_i z_i (\hat{N_i}-w_i/2) }{\hat{N}^2\hat{\mu}}-1,
\end{align*}
with $w_i$ denoting sampling weights attached to unit $i$, $\hat{N}=\sum_{k \in \mathcal{S}} w_k $, $\hat{\mu}=\sum_{k \in \mathcal{S}} w_k z_k/\hat{N}$, and $\hat{N}_i=\sum_{k \in \mathcal{S}} w_k \mathds{1}(r_k \leq r_i)$. The notation $\mathds{1}(A)$ defines an indicator function, assuming value 1 if $A$ is observed and 0 otherwise.
The weights could be the inverse of the inclusion probabilities or a treated and calibrated version of them.
We adopted its bias-corrected version proposed by \citet{Denicolo2021} as follows
\begin{align*} 
G_{adj}= \frac{\tilde{n}}{\tilde{n}-2}\bigg[G_w- \frac{2\hat{\gamma}}{\hat{\mu}^3} \mathbb{V}[\hat{\mu}]+ \frac{2}{\hat{\mu}^2} \text{Cov}[\hat{\mu}, \hat{\gamma}]\bigg],
 \end{align*}
with  $\tilde{n}=\sum_{k \in \mathcal{S}} \mathds{1}(w_k \neq 0)$ and
$
\hat{\gamma}= \sum_{i \in \mathcal{S}} w_i z_i (\hat{N_i}-w_i/2) /\hat{N}^2$. $\mathbb{V}[\hat{\mu}]$ and $\text{Cov}[\hat{\mu}, \hat{\gamma}]$ denote variance and covariance of design estimators.

Despite its fame, the Gini index has some drawbacks. First of all, it is a stochastic dominance measure enabling only for partial ordering of probability distributions. Namely, this index is able to determine which distribution precedes the other in the ordering only among certain pairs of probability distributions. Secondly, it does not allow for decomposability into within and between components. Thirdly, it is weakly (positional) transfer sensitive which means that, in case of income transfers, the index varies depending on the donor and recipients ranks.

The Relative Theil index, instead, is additive decomposable and has the advantage to be strongly transfer-sensitive, meaning that the measure reacts to transfers depending on the donor and recipient income levels. It is an entropy-based measure and is set up as the relative formulation of the more famous Theil index \citep{theil1979world}, i.e. scaled on the maximum of its support ($\log n$). Its estimator in the $srs$ case is defined as follows
\begin{align}
R= {\frac  {1}{n \log(n)  }}\sum _{i \in \mathcal{S}_{iid}} {\frac  {z_{{i}}}{\hat{\mu}_{srs}}}\log \bigg( {\frac  {z_{{i}}}{\hat{\mu}}} \bigg).
\label{srsrt}
\end{align}
In the complex survey case, the Horwitz-Thompson type estimator for the Theil index has been considered in its bias-adjusted formulation \citep{Biewen2006, Denicolo2021}. This has been adapted to the relative case by replacing the superior bound of its support with its population value, in order to not induce further bias, as follows
\begin{align*}
T_{adj}&= \frac{1}{ \hat{N} } \sum_{i \in \mathcal{S}} w_i  \frac{z_i}{\hat{\mu}} \log \frac{z_i}{\hat{\mu}} +\frac{ \text{Cov}[\hat{\mu}, \hat{\varpi}]}{ \hat{\mu}^{2}} -\bigg( \frac{\hat{\varpi}}{\hat{\mu}^{3}}+\frac{1}{2\hat{\mu}^{2}} \bigg)\mathbb{V}[\hat{\mu}]\\
  R_{adj}&=\frac{T_{adj}}{\log N}
\end{align*}
with $\hat{\varpi}= \sum_{i \in \mathcal{S}} w_i z_i \log z_i/ \hat{N} $.

Another perspective on inequality is depicted by the family of Atkinson Indexes \citep{atkinson1970measurement}. They provide for an explicit value judgment by incorporating in the measurement a social welfare function, regulated by a parameter $\varepsilon$. Under this normative approach, the index value has clear meaning, quantifying the amount of welfare loss of the current inequality level: a value of 0.30 means that “if incomes were equally distributed then we should need only the 70\% of the present national income to achieve the same level of social welfare” \citep{atkinson1970measurement}. They can be seen, therefore, as measures of distributional inefficiency.
Moreover, they satisfy a multiplicative decomposition property \citep{la2008extended} and may provide for a complete ranking among alternative distributions, i.e. being able to determine the ordering among every pair of distributions, and thus to establish a ranking among the full set of distributions.
This happens at the expense of more stringent and subjective assumptions on the choice of the welfare utility function to adopt \citep{bellu2006policy}. Under a concave utility function (whose concavity level is regulated by $\varepsilon$), the estimator of Atkinson index in the $srs$ case is defined as
\begin{align}  \label{srsatk1}
A(\varepsilon \neq 1)&=1-{\frac {1}{\hat{\mu}_{srs} }}\left({\frac {1}{n}}\sum _{i \in \mathcal{S}_{iid}} z_{i}^{1-\varepsilon }\right)^{1/(1-\varepsilon )} \\
A(\varepsilon =1)&=1-{\frac {1}{\hat{\mu}_{srs}}}\exp \bigg\lbrace \frac{ \sum_{i \in \mathcal{S}_{iid}} \log z_i }{n}\bigg\rbrace
 \label{srsatk2}
\end{align}
with  $\varepsilon \geq 0$.
 The parameter $\varepsilon$ denotes the level of inequality aversion: at increasing values of $\varepsilon$, the index becomes more sensitive to changes at the lower end of the income distribution and vice versa. We consider specifically the two indexes referring to $\varepsilon=\lbrace 0.5, 1 \rbrace$ values, which incorporate nice robustness properties to outliers, showing at the same time different sensitivities \citep{Denicolo2021}.
The estimator referred to complex survey case \citep{Biewen2006} is
\begin{align*}
A_w(\varepsilon\neq 1 )&=1-\frac{1}{\hat{\mu}} \bigg( \frac{1}{\hat{N}} \sum_{i \in \mathcal{S}} w_i z_i^{1-\varepsilon}\bigg)^{1/(1-\varepsilon)}\\
A_w(\varepsilon=1)&=1-  \frac{1}{\hat{\mu}} \exp \bigg\lbrace \frac{ \sum_{i \in \mathcal{S}} w_i \log z_i }{\hat{N}}\bigg\rbrace.
\end{align*}
We adopted their bias-adjusted versions \citep{Denicolo2021} as follows
\begin{align*}
A_{adj}(\varepsilon \neq 1) &=A_w(\varepsilon)+[1-A_w(\varepsilon)] \times \\
& \quad \times \bigg[
\frac{\varepsilon \cdot \mathbb{V}[\hat{\varrho}]}{2(1-\varepsilon)^2} [\hat{\mu}-\hat{\mu} A_w(\varepsilon)]^{2\varepsilon-2}+ \frac{\mathbb{V}[\hat{\mu}]}{\hat{\mu}^2}
-\frac{\text{Cov}[\hat{\varrho}, \hat{\mu}]}{\hat{\mu}^{2-\varepsilon}(1-\varepsilon)} [1-A_w(\varepsilon)]^{\varepsilon-1}  \bigg] \nonumber\\
A_{adj}(\varepsilon =1 )&=A_w(\varepsilon)+[1-A_w(\varepsilon)]\bigg[ \frac{\mathbb{V}[\hat{\iota}]}{2} +\frac{\mathbb{V}[\hat{\mu}]}{ \hat{\mu}^{2}}  -\frac{\text{Cov}[\hat{\iota}, \hat{\mu}]}{ \hat{\mu}}\bigg]
 \end{align*}
with $\hat{\varrho}= \sum_{i \in \mathcal{S}} w_i z_i^{1-\varepsilon}/  \hat{N} $ and  $\hat{\iota}= \sum_{i \in \mathcal{S}} w_i \log z_i/ \hat{N}$.

\subsection{Variance Estimation}
\label{variancestimation}

The sampling variances of complex survey estimators have been estimated from the data following a two steps strategy as in \citet{Fabrizi2011}. As a first step, we performed a proper bootstrap procedure developed taking into account the complex sampling design \citep{fabrizi2020functional, Denicolo2021}, using 1,000 bootstrap samples.
Secondly, the bootstrap estimates have been smoothed via a Generalized Variance Function (GVF) approach in order to mitigate the uncertainty of sampling variances induced by small sample sizes.

The definition of a GVF smoothing model needs assumptions on the shape of the variance function for such inequality estimators.
Thus, in the spirit of what was done by \citet{Fabrizi2016} for the Gini index, we derived the variance function of Relative Theil and Atkinson index (for any $\varepsilon)$ under specific simplifying conditions, such as the usual log normality assumption of income variable and $srs$ setting. Indeed, closed-form results for the complex survey case are hard to obtain for non-linear statistics, such as the ones we are dealing with, due to the correlation among observations.
The Gini index result, as well as our following derivations for the other measures, has been directly incorporated into the GVF model.

Before moving to the variance function derivation, let us introduce the partition of the population $\mathcal{U}$ into $D$ small domains, such as each domain $d$ has population size $N_d$, with $N=\sum_{d=1}^D N_d$, and samples $\mathcal{S}_{iid}$ and $\mathcal{S}$ are subsequently partitioned into $D$ subsamples of size $n_d$ and $\tilde{n}_d$ respectively, for $d=1, \dots, D$, with $n=\sum_{d=1}^D n_d$ and $\tilde{n}=\sum_{d=1}^D \tilde{n}_d$.

\begin{proposition}
Under the assumption of log-normality of income,  let us consider the $j$-th individual in domain $d$, whose level of income is $z_{jd}$ variable. As a consequence, $\log(z_{jd}) \sim \mathcal{N}(\mu_d, \varphi^2_d)$, $iid$  at varying \textbf{$j=1, \dots, n_d$}. The srs estimator of Relative Theil Index \eqref{srsrt}, for domain $d$, $R_d$ has variance function
\begin{align} \label{re1eq}
\mathbb{V}[R_{d}] \cong \frac{2{\theta_d^R}^2}{n_d},
\end{align}
where $\theta_{d}^R$ denotes its population value.
\label{re1}
\end{proposition}

\begin{proof}
The Relative Theil index is defined for a log-normal income variable as
\begin{align} \label{theil}
\theta_{d}^R=\frac{1}{ \log (n_d)}\bigg( \frac{\mathbb{E}[z \cdot  \log(z)]}{\mathbb{E}[z]}- \log (\mathbb{E}[z]) \bigg) =\frac{ \varphi_d^2}{2  \log(n_d)}.
\end{align}
Since the moments involved in the previous expression are
\begin{align}
\mathbb{E}[z]&=\exp \bigg\lbrace \mu_d+\frac{\varphi_d^2}{2} \bigg\rbrace \nonumber\\
\mathbb{E}[z \cdot \log(z)]&= \int_0^{+\infty} z \log(z) \frac{1}{z \sqrt{2\pi \varphi_d^2}} \exp \bigg\lbrace -\frac{[\log(z)-\mu_d]^2}{2\varphi_d^2}\bigg\rbrace dx \nonumber \\
&=   \int_{-\infty}^{+\infty}  t  \exp\{t\} \frac{1}{\sqrt{2\pi \varphi_d^2}} \exp \bigg\lbrace -\frac{(t-\mu_d)^2}{2\varphi_d^2}\bigg\rbrace dt \nonumber \\
&=  (\varphi^2_d+\mu_d) \exp \bigg \lbrace\mu_d+ \frac{\varphi_d^2}{2} \bigg\rbrace, \label{integrale}
\end{align}
where the last step \eqref{integrale} involves equation 3.462.6 in \citet{gradshteyn2014table}.
Considering that $\varphi_d^2$ is estimated by
$s_d^2=\frac{1}{n_d-1} \sum_{j=1}^{n_d} [\log(z_{jd})-\hat{\mu}_d]^2$ and $\mathbb{V}[s_d]\cong \frac{\varphi_d^2}{2n_d}$, by applying delta method the result follows
\begin{align}
\mathbb{V}[R_{d}]=\mathbb{V}\bigg[\frac{s_d^2}{2 \log(n_d)}\bigg] & \cong \frac{\varphi_d^4}{2 \log^2(n_d) n_d} \nonumber \\ & = \frac{2{\theta_d^R}^2}{n_d} \label{varelast}
\end{align}
where equation (\ref{varelast}) is obtained by (\ref{theil}) considering that $\varphi_d^2=2 \theta_{d}^R \log(n_d)$.
\end{proof}

\begin{proposition}
Under Proposition \ref{re1} assumptions, the srs estimators of Atkinson index \eqref{srsatk1} and \eqref{srsatk2} for domain $d$, denoted with $A_d(\varepsilon)$, have variance function
\begin{align} \label{varatk2}
\mathbb{V}[A_d(\varepsilon)] \cong \frac{2\theta_{d}^{A}(\varepsilon)^2}{n_d} \exp\lbrace -2\theta_{d}^A(\varepsilon) \rbrace,
\end{align}
where $\theta_{d}^A(\varepsilon)$ denotes the population value of the index.
\label{atk}
\end{proposition}

\begin{proof}
The population value of Atkinson index in domain $d$ under the mentioned assumptions, for any $\varepsilon \geq 0$ and $\neq 1$, is
\begin{align} \label{atk2}
\theta_{d}^A(\varepsilon)=1-\exp \lbrace-\varepsilon  \varphi_d^2/2 \rbrace,
\end{align}
with $\varphi_d^2$ estimated by
$s_d^2$. By applying the normal distribution theory, $\mathbb{V}[s_d]\cong \frac{\varphi_d^2}{2n_d}$ and using the delta method:
\begin{align}
\mathbb{V}[A_d(\varepsilon)]=\mathbb{V}\bigg[1-\exp \bigg\lbrace- \frac{\varepsilon s_d^2}{2}   \bigg\rbrace \bigg] & \cong \frac{\varepsilon^2 \varphi_d^4}{2 n_d}\exp\lbrace -\varepsilon \varphi_d^2\rbrace \nonumber\\
&\cong \frac{2\theta_d^A(\varepsilon)^2}{n_d} \exp\lbrace -2\theta_d^A(\varepsilon) \rbrace,
\label{varatklast}
\end{align}
where equation (\ref{varatklast}) is obtained by Maclaurin expanding (\ref{atk2}), so that $\varphi_d^2 \cong 2\theta_{d}^A(\varepsilon)/\varepsilon$. Note that this result can be easily generalized to the case $\varepsilon=1$.
\end{proof}

Since Proposition \ref{re1} and \ref{atk} has been derived under specific simplifying assumptions, we need to ensure their validity in our context. Thus, we tested them on our bias-corrected estimators for complex survey, namely $A_{adj}$, $G_{adj}$ and $R_{adj}$, through a Monte Carlo simulation. We used EU-SILC data, described in detail in Section \ref{application}, as synthetic population by considering 21 NUTS-2 Italian regions as domains of interest. In order to circumvent possible null or negative income values and non-robustness to outliers, we treated income data by using a semi-parametric Pareto and inverse-Pareto tail modelling procedure using the Probability Integral Transform Statistic Estimator (PITSE) proposed by \citet{Finkelstein2006} and \citet{masseran2019power}.
We draw 1,000 samples by mimicking EU-SILC complex scheme, stratified with two-stage selection. Then we compared Monte Carlo variances with Proposition \ref{atk} and \ref{re1} results, showing very high correlations: 0.79 for Gini index, 0.92 for Atk(1), 0.86 for Atk(0.5) and 0.99 for Relative Theil.
The empirical results, as expected, clearly underestimate the variances in comparison with the Monte Carlo ones, since the effect of the design is ignored. However, the relationship is strongly linear and by fitting a regression with Monte Carlo variances as response and the empirical ones as explanatory, intercepts are zeros and slopes end up to be 2.04 for Relative Theil, 2.23 for Atk(0.5), 2.27 for Atk(1) and 4.68 for Gini index.
The strong linear dependence and proportionality and, at the same time, the non-negligible underestimation, lead us to consider appropriate the implementation of a GVF model.

In the following, the GVF model setting is unravelled, by considering also the Gini index variance derived by \citet{Fabrizi2016} under the same assumptions of Propositions \ref{re1} and \ref{atk}, defined approximately for domain $d$ as
 \begin{align} \label{gini}
\mathbb{V}(G_d) \cong \frac{{\theta_d^G}^2 (1-{\theta_d^G}^2)}{n_d},
  \end{align}
with $\theta_d^{G}$ its population value.

Consider that under a complex survey scheme, the sample $\mathcal{S}$ may be less informative than a sample of the same size $\tilde{n}_d$ under srs, being $\mathcal{S}$ affected by dependency across observations.
The effective sample size, i.e. the srs equivalent sample size, is a proxy of the information carried by the sample and can be estimated for $\mathcal{S}$. Let us suppose it corresponds to $\psi_{\text{ind}} \cdot \tilde{n}_{d}$ for any considered index $ind$, with $\psi_{\text{ind}} > 0$ denoting a deflating factor induced by the dependence. The latter quantity can be alternatively defined as 
the inverse of the design effect deff$_{\text{ind}}$, i.e. the ratio between the design-based variance of a generic index estimator and its srs variance, which measures the amount of variance inflation induced by the complex selection process. Note that the design effect was assumed constant across areas in order to guarantee the smoothing of
original estimates. \textcolor{black}{Let denote with $\hat{\mathbb{V}}[\cdot ]_{boot}$ a bootstrap estimator with large error}, being $y_d$ a generic inequality estimator, among $A_{adj}(\varepsilon)$, $G_{adj}$ and $R_{adj}$, defined for domain $d$ whose population values is $\theta_d$. The numerator of its variance function is generically defined by $f(\theta_d)$. A GVF model is set up by assuming that
\begin{align*}
\mathbb{V}[y_d] = \frac{f(\theta_d)}{\psi_{\text{ind}} \cdot  \tilde{n}_d}.
\end{align*}
 Therefore, we introduce the following smoothing model estimated via generalized least squares
\begin{align*}
 \frac{f(y_d)}{\hat{\mathbb{V}}[y_d]_{boot}} = \psi_{\text{ind}} \tilde{n}_d+\epsilon_d,
\end{align*}
\noindent where $\epsilon_d$ denotes zero-mean heteroskedastic residuals.
The smoothed estimator comes from \eqref{varatk2},  \eqref{re1eq}, \eqref{gini} by replacing $\theta_{d}$ with $y_d$ and $n_d$ with $\tilde{n}_d \cdot \hat{\psi}_{\text{ind}}$, where $\hat{\psi}_{\text{ind}}$ is the GLS estimate of $\psi_{\text{ind}}$.
The pseudo $R^2$ for the smoothing models are respectively, 0.78 for Gini index, 0.72 for Atkinson ($\varepsilon=1$) index, 0.67 for Atkinson ($\varepsilon=0.5$) index and 0.48 for the Relative Theil index. The latter result is due to the instability of the ratio $ f(y_d)/\hat{\mathbb{V}}[y_d]_{boot}$ in case of values close to zero of both the numerator and the denominator.

   \begin{figure}[H]
\centering
    \includegraphics[width=\textwidth]{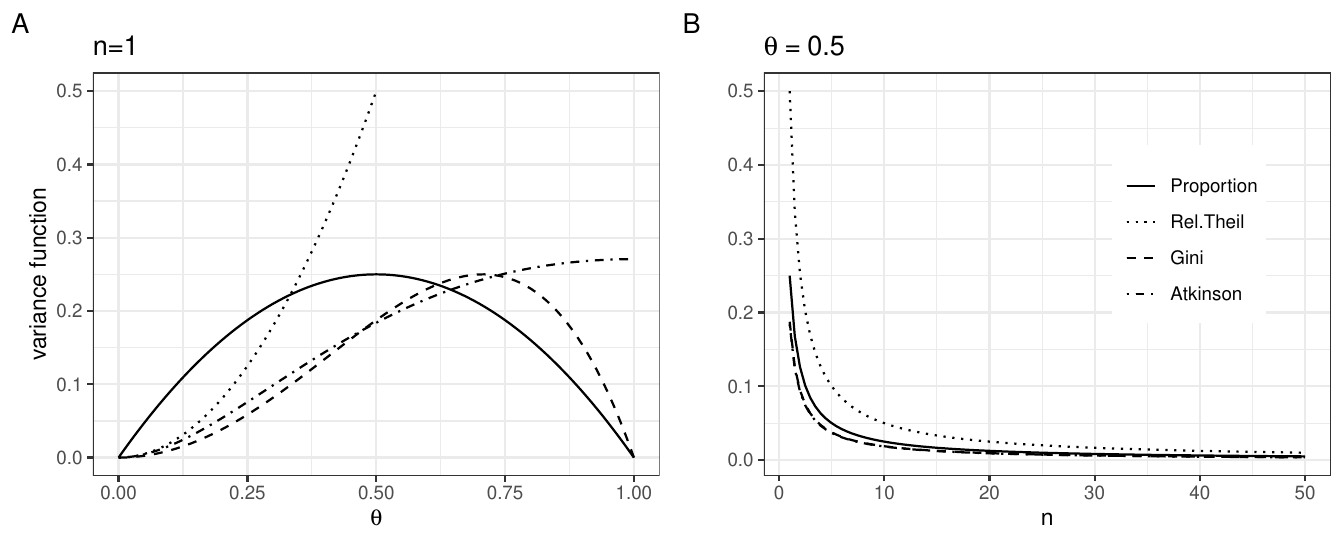}
  \caption{Variance functions with respect to $\theta_d$ ($n_d=1$) on the left-hand side, with respect to $n_d$ ($\theta_d=0.5$) on the right-hand side, for each measure and considering the proportion case.}
  \label{varfun}
  \end{figure}

 Results of Propositions \ref{re1} and \ref{atk} show a very different structure for the variance function of Atkinson and Relative Theil indexes with respect to the Gini index and proportions. A comparison plot can be found in Figure \ref{varfun} displaying each variance function with respect to $\theta_d$ and $n_d$. As opposed to the proportion and Gini index cases, the Atkinson and Relative Theil variance functions are both monotonically increasing, as clear from (\ref{atk2}) and (\ref{theil}). Thus, greater values of $\theta_d$ correspond to greater log income dispersion, inevitably leading to an increase in index variability. The explosive trend of Relative Theil variance is related to its explosive connection with the log income population variance (\ref{theil}). Moreover, notice that the variance function of Atkinson index in (\ref{varatk2}) does not directly depend on its parameter $\varepsilon$, being fixed for the whole parametric family. This does not happen for the Generalized Entropy parametric family, as shown in Section S.1 of the Supplementary Material.

The precision estimates obtained for NUTS-3 Italian regions, employing EU-SILC data (described in Section \ref{application}), are analyzed in terms of the coefficient of variation (CV). ISTAT guidelines state that CV should not exceed 15\% for domains and 18\% for small domains in case of released estimates, otherwise, this serves as an indication to perform small area estimation \citep{eurostat2013handbook}.
 In our case, Relative Theil and Atkinson ($\varepsilon= 0.5, 1 $) indexes show very similar CV distributions with medians slightly lower than 18\%, ranging totally from 6\% to 54\%. This means that half of the domains has non-reliable estimates. On the other hand, the CV of the Gini index is significantly lower, confirming its robustness to outliers, ranging from 3\% to 27\%, with 7 domains having out-of-bound CV. This further motivates us to employ a small area model.

\section{Small Area Models}
\label{smallareamodels}

We propose a Beta mixture model for small area estimation which from now on is named Flexible Beta (FB) model as in Migliorati et al. (2018). In order to evaluate its performance, we compare the estimation results with those obtained by the well-known Beta small area model.
We start describing the Beta model in Subsection \ref{betamodel}, then we set out our proposal in Subsection \ref{flexbetamodel}. Both models are completed by the prior setting described in Subsection \ref{priordistribution} and the Bayesian estimation detailed in Subsection \ref{modelestimation}.

\subsection{The Beta Model}
\label{betamodel}

Let us consider the Beta distribution with mean-precision parametrization \citep{ferrari2004beta}, such that a Beta distributed random variable is denoted with $Y \sim Beta(\mu \phi, (1-\mu)\phi)$, and has probability density function
\begin{align*}
    f_{B}(y; \mu, \phi)= \frac{\Gamma[\phi]}{\Gamma[\mu\phi] \Gamma[(1-\mu)\phi]} y^{\mu\phi-1} (1-y)^{(1-\mu)\phi-1},  \quad 0<y<1.
\end{align*}
Mean and variance are respectively
\begin{align}
\mathbb{E}[Y]=\mu,
\quad \quad \quad
\mathbb{V}[Y]=\frac{\mu(1-\mu)}{\phi+1},
\label{betavar}
\end{align}
with $0<\mu<1$ and $\phi>0$. A classical Beta small area model for $y_d$, denoting the direct estimator of a generic inequality measure and $\boldsymbol{x}_d$ a set of $P$ covariates for domain $d$, constitutes as a hierarchical model with two levels. At the sampling level, the conditional distribution of the direct estimator is modeled as
\begin{align*}
 y_d | \theta_d, \phi_d \stackrel{ind}{\sim} Beta(\theta_d \phi_d, (1-\theta_d)\phi_d), \quad \forall d.
\end{align*}
In this case, $\mathbb{E} [y_d | \theta_d, \phi_d ]=\theta_d$ is the target parameter and is estimated via a logit regression at the linking level, i.e. $\text{logit}(\theta_d) | \boldsymbol{\beta}, v_d =\boldsymbol{x}_d^T\boldsymbol{\beta}+v_d$, with $v_d | \sigma_v^2 \stackrel{ind}{\sim} \mathcal{N}(0, \sigma_v^2)$ being an area specific random effect.

In literature, small area Beta models assume $\phi_d$ as known, in parallel with the known sampling variance assumption of the classical Fay-Herriot model, in order to allow identifiability. Being usually employed for proportions, this parametrization extremely simplifies the posterior geometry, given that, in this case, the variance structure is $\mathbb{V} [y_d | \theta_d, \phi_d]=\theta_d(1-\theta_d)/n_{d\hspace{1mm}}$ under a binomial process. Thus, by combining it with (\ref{betavar}), $\phi_d+1$ can be seen as the effective sample size under a complex survey scheme. Its estimation is carried out considering the design effect, 
so that $\phi_d+1 = \tilde{n}_d \psi_{\text{ind}} = \tilde{n}_d/ \text{deff}_{\text{ind}} $.

Different approaches have been adopted in small area context for the estimation of $\text{deff}_{\text{ind}}$:
\begin{itemize}
    \item estimation as a unique parameter across all areas within the hierarchical model, such as in \citet{bauder2015small} and \citet{Souza2016} in case of proportions,
     \item separate estimation via a variance smoothing model of a common parameter across all areas, as in \citet{Fabrizi2011,fabrizi2016hierarchical} and \citet{Fabrizi2016} similar to the one applied in Subsection \ref{variancestimation},
    \item separate estimation of a set of parameters varying across all areas through the methodology proposed by \citet{kish1992weighting}, based only on design weights, as in \citet{Wieczorek2012} and \citet{liu2007hierarchical}. \citet{kalton2005estimating} found this approximation reasonably accurate for proportions between 0.2 and 0.8.
\end{itemize}

We decided to tackle the problem from a different perspective, by assuming the sampling variance as known, rather than $\phi_d$, and we estimate it separately via a two-step procedure as in Subsection \ref{variancestimation}. This decision has been taken for several reasons. Firstly, a direct estimation of $\phi_d$ within the model appears cumbersome to carry on, due to the complex structure of the variance functions defined in Propositions \ref{re1} and \ref{atk}, leading to a tricky and intractable parametrization. In second place, the known variance assumption is a standard approach across different small area models (e.g. the Gaussian ones). This preserves the set of assumptions and data inputs across different models, favouring consistency of diagnostic measures, such as the goodness-of-fit ones, and allowing for performance comparison and model selection.

\subsection{The Flexible Beta Model}
\label{flexbetamodel}

The Flexible Beta distribution, introduced by \citet{Migliorati2018}, is a mixture of two Beta random variables with different locations and a common dispersion parameter. Its probability density function is
\begin{align*}
f_{FB}(\lambda_1, \lambda_2, \phi, p)= p \cdot f_{B}(y; \lambda_1, \phi)+ (1-p) \cdot f_{B}(y; \lambda_2, \phi),
\end{align*}
with $0<\lambda_2<\lambda_1<1$ distinct ordered means, in order to avoid label switching problems, $0<p<1$ mixing coefficient and $\phi$ common dispersion parameter. This mixture extends the variety of
shapes of Beta distribution in terms of bimodality, asymmetry and tail behavior. Besides, it ensures that each component is distinguishable, being computationally tractable. 

Our small area model proposal for $y_d$ includes, at sampling level, the Flexible Beta as likelihood assumption:
\begin{align*}
y_d|  \lambda_{1d}, \lambda_{2d}, \phi_d, p \stackrel{ind}{\sim} FB (\lambda_{1d}, \lambda_{2d}, \phi_d, p), \quad  \forall d.
\end{align*}
In this case, the expected value and dispersion parameters of the mixture components vary across areas, while the mixing proportion $p$ remains fixed.
Let us denote with $\boldsymbol{\eta}$ the entire set of parameters, namely $\boldsymbol{\eta}=(\lambda_{1d}, \lambda_{2d}, \phi_d, p)$.
In line with \citet{Migliorati2018}, the parametrization considered in order to carry on estimation is
 \begin{align*}
 y_d| \boldsymbol{\eta} \sim FB (\tilde{w}_d + \lambda_{2d}, \lambda_{2d}, \phi_d, p),
 \end{align*}
 with $\tilde{w}_d=\lambda_{1d}-\lambda_{2d}>0$ denoting the distance between mixture components.  Under such model, the expected value and variance are defined respectively as
\begin{align}
\mathbb{E}[y_d | \boldsymbol{\eta}]&=\theta_d = \lambda_{2d} + p \cdot \tilde{w}_d, \label{thetadef}\\
\mathbb{V}[y_d | \boldsymbol{\eta}]&= \frac{\theta_d(1-\theta_d) +  p(1-p) \tilde{w}_d^2 \phi_d} {\phi_d+1}.\label{varfb}
\end{align}

At the linking level, we model the mean of the lowest component with a logit regression, by preserving the Gaussian random effect assumption, as follows
\begin{align} \label{linkinglevel}
&\text{logit}(\lambda_{2d}) | \boldsymbol{\beta}, v_d = \boldsymbol{x}_d^T \boldsymbol{\beta} +v_d \\
&v_d | \sigma_v^2 \stackrel{ind}{\sim} \mathcal{N}(0, \sigma^2_v)\quad \forall d.
\nonumber
\end{align}
As opposed to the FB regression proposed by \citet{Migliorati2018} and to the classical Beta regression, the linear predictor does not model directly the mean but rather a mixture component mean $\lambda_{2d}$, which in this case can be seen as a pure location parameter.

Being $\theta_d$ our parameter of interest, we assume it as a result of the combination of a location component and a deviation from it, caused by the intrinsic skewness of the sampling distribution, as in \eqref{thetadef}.
Since we are generally dealing with right-skewed distributions, we assume the lower mixture component mean as the pure location parameter ($\lambda_{2d}$) and parameters $p$ and $\tilde{w}_d$ as the ones able to capture such deviations.
If $\theta_d$ was modelled through a logit regression, the relation among parameters would imply $\theta_d=\text{logit}^{-1}(\boldsymbol{x}_d^T \boldsymbol{\beta} +v_d)$, letting the linking level parameters masking such effect. On the other hand, in our case, we consider the location $\lambda_{2d}$ to be directly modelled at the linking level, separating \eqref{thetadef} and \eqref{linkinglevel} and letting $p$ and $\tilde{w}_d$ free to account for area-specific deviations.
In this way, our location-modelling approach unleashes $\theta_d$ estimation.
This is confirmed by the fact that, when $\theta_d$ is rigidly modelled through a logit regression, we notice that its estimates are basically overlapping the ones of the Beta model in Section \ref{betamodel} and estimation time is higher.
The modelling of a location parameter different from the mean at the linking level is well-established in small area literature, we recall zero or zero/one inflated Beta \citep{ Wieczorek2012, fabrizi2016hierarchical}, and skew-normal models \citep{Ferraz2012, ferrante2017small}.

  Similarly to the Beta model, the sampling variance $\mathbb{V}[y_d|  \boldsymbol{\eta}]$ is assumed to be known and replaced by a refined estimate $\hat{\mathbb{V}}[y_d]$, as shown in Subsection \ref{variancestimation}.
 The conditioning is not emphasized on the refined estimate to underline the fact that it is not a model estimate but rather an independent one (survey estimate), treated as a known value in a small area model. As a consequence, the dispersion parameter is not directly estimated and can be obtained from \eqref{varfb} as
\begin{align} \label{sigma}
\phi_d& | \theta_d, p, \tilde{w}_d = \frac{\theta_d(1-\theta_d)-\mathbb{V}[y_d |  \boldsymbol{\eta}]}{\mathbb{V}[y_d |  \boldsymbol{\eta}]-p(1-p)\tilde{w}_d^2}.
\end{align}

 Since estimation requires a variation-independent parametrization, we decided to leave $\lambda_{2d}, \phi_d$, and $p$ free to assume any value of their support and to constrain $\tilde{w}_d$, as in the following Proposition.
\begin{proposition}
 Under FB model and the assumptions of Proposition \ref{atk}, let us consider $\phi_d$ and its relation with the other parameters defined in (\ref{sigma}). In order to preserve its bounded support, i.e. $\phi_d>0$,  $\tilde{w}_d$ has to be constrained such that
 \begin{align}
  \begin{cases}
 \tilde{w}_d<\sqrt{\frac{\mathbb{V}[y_d |  \boldsymbol{\eta}]}{p(1-p)}} \quad \text{if} \quad \theta_d<c \\
  \tilde{w}_d>\sqrt{\frac{\mathbb{V}[y_d |  \boldsymbol{\eta}]}{p(1-p)}} \quad \text{if} \quad \theta_d>c
  \end{cases}
  \label{constrains}
    \end{align}
    with $c$ being a threshold that varies according to $n_d$ and the measure considered: is equal to $n_d/(n_d+2)$ for Relative Theil index, $1/2 \times (\sqrt{4n_d+1}-1)$ for Gini index and does not have closed form for Atkinson index.
    \label{prop3}
\end{proposition}

\begin{proof}
By imposing $\phi_d > 0$ on (\ref{sigma}), the solution comes out to be
\begin{align}
 \mathbb{V}[y_d |  \boldsymbol{\eta}] \in
 \bigg( \min \bigg\{ \theta_d(1-\theta_d),  p(1-p)\tilde{w}_d^2 \bigg\},
 \max \bigg\{ \theta_d(1-\theta_d),  p(1-p)\tilde{w}_d^2 \bigg\} \bigg) . \label{mycase}
  \end{align}
After splitting the problem into two cases, namely
 \begin{align}
 \mathbb{V}[y_d |  \boldsymbol{\eta}] <  \theta_d(1-\theta_d) \quad\quad \text{and}\quad\quad  \mathbb{V}[y_d |  \boldsymbol{\eta}] >  \theta_d(1-\theta_d)
 \label{caso}
\end{align}
we substitute $\mathbb{V}[y_d |  \boldsymbol{\eta}]$ in \eqref{caso} with (\ref{varatk2}), (\ref{re1eq}) and (\ref{gini}). The results for all three measures have respectively the following shapes $\theta_d<c$ and  $\theta_d>c$, where $c$ depends on $n_d$ and differs for any measure. Then, we solve equation \eqref{mycase} for $\tilde{w}_d$ for each case obtaining \eqref{constrains}.
\end{proof}

To understand the behaviour of previous constraints within our inferential problem, we evaluated them by considering the case $n_d=2$, as a degenerate case with maximum variance. Indeed, a lower sample size does not allow for inequality measurement. We numerically derive the minimum of $c$, being $0.50$ for the Relative Theil index, $0.84$ for Atkinson indexes and $1$ for the Gini index.
Discarding the Gini index being always $\theta^G_d<1$, the case $\theta_d > c$ is totally implausible for all the other income inequality measures. Indeed, those values correspond to far-fetched values of log income variable dispersion that, following (\ref{atk2}) and (\ref{theil}), equals to $\varphi_d^2>0.69$ for Relative Theil index, $\varphi_d^2>3.71$ for Atkinson indexes. To be clear, a log-normal fitting on 2017 EU-SILC equivalent disposable income done by \citet{Denicolo2021}, shows $\hat{\varphi}^2=0.18$.
Therefore, we opt to consider only the case $\theta_d<c$ in Proposition \ref{prop3}; the same reasoning could be easily done in case of proportions, where $\theta_d<c$ holds for any $n_d>1$.
An evaluation of the admissible support of $\theta_d$ in the considered case at varying $n_d$ can be found in Section S.3 of the Supplementary Material.

As a consequence, the range of $ \tilde{w}_d$ is defined as
  \begin{align}
 \tilde{w}_d \in \bigg (0, \min  \bigg\lbrace \frac{1-\lambda_{2d}}{p}, \sqrt{\frac{\mathbb{V}[y_d |  \boldsymbol{\eta}]}{p(1-p)}}\bigg\rbrace \bigg), \label{secondvinculum}
   \end{align}
where the upper bound of its support has a double vinculum. The first term, on $\lambda_{2d}$, can be seen as a support vinculum, since it allows $\theta_d$ to be upper bounded, i.e. $\theta_d<1$. The second one follows from Proposition \ref{prop3} and clearly takes into account the fact that the imposed sampling variance (given $p$) has to constrain the distance between mixture components. Thus, the distance has been modelled in line with (\ref{secondvinculum}) as
\begin{align*}
\tilde{w}_d=w \cdot \max{\text{supp}} \lbrace \tilde{w}_d \rbrace= w \cdot \min  \bigg\lbrace \frac{1-\lambda_{2d}}{p}, \sqrt{\frac{\mathbb{V}[y_d |  \boldsymbol{\eta}]}{p(1-p)}} \bigg\rbrace
\end{align*}
unleashing parameter $w$ free to vary in $(0,1)$, being common across the areas. The underlying assumption implies that, for any given domain, different determinations of each direct estimator pertain to two latent groups, one of which displays a greater mean than the other, for any given set of covariates. Parameter $w$ retains the meaning of distance between the regression functions of the two groups and it can be called the normalized distance \citep{Migliorati2018}.

The underlined parametrization is variation independent without penalizing the interpretability of the parameters.
Moreover, the target parameter defined in \eqref{thetadef} can be rewritten as
\begin{align}
\theta_d |\lambda_{2d}, p, w  =\lambda_{2d} + p \cdot w \cdot \min  \bigg\lbrace \frac{1-\lambda_{2d}}{p}, \sqrt{\frac{\mathbb{V}[y_d |  \boldsymbol{\eta}]}{p(1-p)}}\bigg\rbrace,
\label{targetp}
\end{align}
being a sum between the location parameter  $\lambda_{2d}$ and a second term depending on $\lambda_{2d}$, the sampling variance, the mixing parameter $p$ and the normalized distance $w$. This shape permits to grasp an alternative interpretation of $w$ as a factor that regulates the impact of sampling variance on $\theta_d$ (given $p$). Indeed, due to the different scale, usually $(1-\lambda_{2d})/p > \sqrt{\mathbb{V}[y_d |  \boldsymbol{\eta}]}/\sqrt{p(1-p)}$. Note that the structure of $\theta_d$ is quite similar to its corresponding parameter in case of skew-normal likelihood  \citep{moura2017small, ferrante2017small}. In this case, the expected value is the sum of the location parameter and another component that depends on the known sampling variance and a skewness parameter.

Given the above considerations, as long as the sampling variances decrease, and presumably area sample sizes increase, the conditional distribution of direct estimator $y_d$  stretches to a Beta distribution:
\begin{align}
y_d | \boldsymbol{\eta} \stackrel{n_d  \to +\infty}{\sim}  Beta(\theta_d \phi_d, (1-\theta_d) \phi_d).
\label{convergencetobeta}
\end{align}
The expected value tends to the linear predictor and $\phi_d  \rightarrow{} \theta_d (1-\theta_d )/\hat{\mathbb{V}}[y_d]-1$ as in (\ref{betavar}).
Moreover, it is well known that a Beta distribution with large shape parameters, i.e. low variance, converges to a normal distribution. Therefore, it is possible to state that, as the sampling variance tends to 0, our sampling model tends asymptotically to the Gaussian one. Also when $p \rightarrow{} 0 $, $p \rightarrow{} 1 $ or  $w \rightarrow{} 0$, (\ref{convergencetobeta}) is verified, as common in degenerate mixture models \citep{fruhwirth2019handbook}.

To summarize, the parameter to be estimated in FB model are the ones related to the linear predictor $\boldsymbol{\beta}$ and the random effect variance $\sigma_v$, the mixing coefficient $p$ and the normalized distance between mixture components $w$. Parameters $p$ and $w$ adjust the predictor depending on the magnitude of its sampling variance, guaranteeing a more flexible mean modelling as shown in \eqref{targetp}. This can be seen as the main characteristic of the FB small area model.

\subsection{Prior Distributions}
\label{priordistribution}

The following weakly-informative priors complete the Beta model:
\begin{align}
\sigma_v  \sim \text{Half-} \mathcal{N} (0, \nu^2) \quad \text{and} \quad \beta_k  \sim \mathcal{N}(0, \sigma^2), \label{betasigmaprior} \quad \forall k=1, \dots, P
\end{align}
 with $\sigma^2=10$. Considering the scale of the logit transformation, $\nu^2=1$ can be seen as quite a non-informative option. For the FB model, the choice of the prior includes in addition to \eqref{betasigmaprior}
\begin{align} \label{priorp}
p \sim \text{Unif}(0,1) \quad \text{and} \quad w \sim \text{Unif}(0,1).
\end{align}
This can be seen as a general formulation in line with \citet{goodrich2018bayesian}. However, in some specific cases e.g. when dealing with a few areas, we recommend using a slightly informative prior for the mixing coefficient to foster convergence or avoid convergence problems, such as
$p \sim Beta(2,2)$. Such an option enables to avoid the boundaries of its support while being very close to a uniform distribution.

\subsection{Model Estimation}
\label{modelestimation}

We estimate the model by adopting a Hierarchical Bayes (HB) approach.
This approach to inference has several benefits in the SAE context \citep[section 10]{rao2015}, as to easily manage non-Gaussian distributional assumptions
and to capture the uncertainty about all target parameters through the posterior distribution.

The FB model falls within the definition of a finite mixture, thus it could be seen as an incomplete data model where the allocation of observations to each mixture component is an unknown and latent component. In this case, a Bayesian approach based on Markov Chain Monte Carlo (MCMC) techniques is particularly suitable for posterior exploration.
 Specifically, the fitting was carried out by implementing the no-U-turn sampler \citep{hoffman2014no}, an adaptive variant of Hamiltonian Monte Carlo (HMC) algorithm via \texttt{Stan} language \citep{carpenter2017stan}.
 The HMC exploits differential geometry properties of the posterior distribution, in order to improve MCMC efficiency \citep{betancourt2017conceptual}.
We performed estimation by using 4 chains, each with 5,000 iterations, discarding the first 2,000 as warm-up.

Within the HB framework, we assume a quadratic loss and define as point predictor of $\theta_d$ its posterior expected value, namely
\begin{align}
\hat{\theta}^{HB}_d=\mathbb{E}[\theta_d|\text{data}] \quad \forall d,
\label{predictor}
\end{align} hereafter named model-based estimate. The posterior variance of the target parameter is used to describe its uncertainty.

An important property of the Fay-Herriot model is that, under the assumption of known random effect variance in HB context, the predictors are the outcome of a shrinkage process in between the direct estimate $y_d$ and the synthetic estimate, being bounded. Predictors tend towards $y_d$ when sampling variance is small in comparison with model variance, and towards synthetic estimate when it goes the other way round. In case of Beta assumption, predictors are not bounded. However, \citet{Janicki2020} proved its asymptotic behavior, showing that it tends towards $y_d$ when sampling variance goes to zero, and towards the synthetic estimate, in this case logit$^{-1}(\boldsymbol{x}_d^T \boldsymbol{\hat{\beta}})$, when model variance goes to zero. The first property, also known as design consistency, has been proved for the Beta model also by \citet{fabrizi2020functional} relying on asymptotic Gaussianity.
Thus, given the asymptotical behavior of our model in \eqref{convergencetobeta}, we can state that the design-consistency property is preserved also under the FB model.

 The FB small area model can be applied to any unit interval response and its estimation has been implemented in the R package \texttt{tipsae} \citep{denicologardini2022}, together with small area specific diagnostics functions and other complementary tools.

\section{Application on EU-SILC Data}
\label{application}

We are interested in estimating inequality in Italian NUTS-3 regions by using EU-SILC data. Given the high level of uncertainty of the direct estimates, described in Subsection \ref{variancestimation}, we employ small area models considering both Beta and FB likelihoods. We estimate four separate univariate models for each likelihood, referring to the four different inequality measures.
A survey data description directly follows, while auxiliary variables, from other sources, are set out in Subsection \ref{auxiliaryvariables}. Eventually, model results are compared in Subsection \ref{results}.

The EU–SILC survey \citep{guio2005income} collects cross-sectional and longitudinal microdata on income, poverty, social exclusion and living conditions in a timely manner. The survey is conducted in each country by the National Institute of Statistics and coordinated by Eurostat, guaranteeing consistent methodology and definitions across all EU member states. The sampling design involves a rotational panel lasting four years, where each year one-quarter of all respondents is newly introduced. As regards the Italian sample provided by ISTAT, the survey units (households) are sampled according to a complex survey scheme involving stratification and two-stage selection. The first-stage units are municipalities, stratified accordingly to the demographic size, the ones with great size are considered as self-representative units and form a take-all stratum. Within the selected primary units, households are drawn randomly as secondary sampling units.

 In our case, we concentrate on the 107 NUTS-3 Italian domains by using the 2017 wave. The sample comprises 22,226 households and 48,819 corresponding individuals.
The domain size ranges from a minimum of 32 to a maximum of 2,536 individuals; with 25th, 50th and 75th percentiles respectively as 196, 314, 612 (from 18 to 1,270 households; with percentiles 86, 138, 275).

\subsection{Auxiliary Variables}
\label{auxiliaryvariables}

The possible determinants of income inequality within European regions have been identified by \citet{perugini2008income}. According to them, the main ones are human capital endowment, labour market performances, economic development, industrial specialization and the demographic structure.

A small area model does not have causal inference ambitions, but rather it requires auxiliary information to be accurately known at population level. Therefore, we restrict the choice to data currently available and measured without error: census and registry office data as well as tax forms data. As a human capital endowment proxy, we calculated the ratio between the number of people aged 15–64 with a high school diploma or higher level of education, and the number of people within the same age class with compulsory education level, based on the 2011 Italian population census data. \citet{Fabrizi2016} refer to this indicator as to the people-in-higher-education ratio. The demographic structure is explained by areal population density and aged dependency ratios. Moreover, as suggested by \citet{perugini2008income}, we used the percentage of resident foreigners (immigrants) and the male/female resident foreigners ratio as indirect measures of economic development.

Concerning fiscal archives data, we include average taxable income claimed by private residents, the percentage of residents aged more than 15 filling tax forms and the percentage of residents with income lower than/greater than double the national median filling tax forms. These variables measure the affluence of income earners in the area and are adopted as indirect proxies of labour market performances \citep{Fabrizi2016}.

Lastly, in order to provide strongly correlated information, we add the corresponding inequality measures calculated on a discrete scale given the income classes declared by tax forms. Note that those measures are estimated on market income (i.e. income before taxes and transfers), while our target variable is the disposable income instead (after taxes and transfers). We obtain raw estimates of market income inequality, legitimately greater than our response due to the redistributive power of taxes and transfers on income distribution. This happens despite the missing component of variability within income classes, not captured by our market income inequality estimators.
All the auxiliary variables were standardized before being incorporated into the models, in order to harmonize the scale, and subjected to preliminary variable selection to avoid multicollinearity.

\subsection{Results}
\label{results}

\begin{table}
    \centering
\begin{tabular}{rrrr}
  \hline
 && Beta & FlexBeta \\
  \hline
   Atk(0.5) &\texttt{looic}  & -572.2 & -595.7 \\
&    (\texttt{se}) & (18.9) & (17.4) \\
&  \texttt{acvr} \% & 43.2 & 51.1\\
Atk(1)&\texttt{looic} & -404.6 & -425.5 \\
&  (\texttt{se}) & (18.6) & (17.7) \\
&  \texttt{acvr} \% & 41.6 & 46.0\\
  Relative Theil& \texttt{looic}   &-966.8&-992.3\\
&  (\texttt{se}) & (16.9)& (15.6) \\
&   \texttt{acvr} \% & 36.9 & 47.3\\
      Gini &  \texttt{looic}  &-388.3&-392.3 \\
&    (\texttt{se}) &(21.3)&(20.5)\\
&  \texttt{acvr} \% & 38.5 & 38.3\\
   \hline
\end{tabular}
\caption{\texttt{looic} and related standard error as well as \texttt{acvr} for each model and each measure}
   \label{looicvr}
\end{table}
\begin{figure}[h]
\centering
    \includegraphics[scale=0.6]{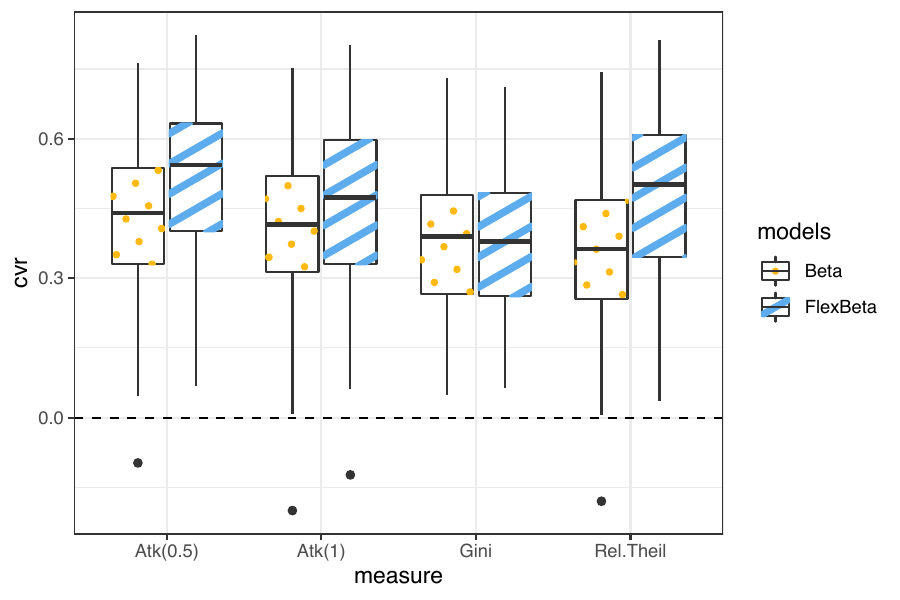}
  \caption{CV reduction for each model and each measure}
  \label{cvred}
  \end{figure}

The model estimation has been carried out and posterior draws have been validated through MCMC diagnostics, showing good chain mixing and quick convergence for any measure.
The estimation times are between 20.62 and 26.90 CPU seconds for the Beta model and between 71.70 and 80.38 CPU seconds for the Flexible Beta model on an Intel Core i5 dual‑core (1,8GHz) processor.
 A model comparison has been also performed through specific model diagnostics. Concerning goodness-of-fit and model comparison, the \texttt{looic} measure, based on leave-one-out cross-validation \citep{Vehtari2017}, and its standard error have been used.
The \texttt{looic} is preferred over the most classical DIC and AIC measures, since it is fully Bayesian, using the entire posterior distribution, it is invariant to parametrization and it works for singular models.

In order to evaluate model-based estimators performances in comparison with direct estimators, the Coefficient of Variation Reduction measure \citep{ferrante2017small} is used to calculate the precision improvement:
\begin{align*}
   \texttt{cvr}_d^{HB}=1-\frac{\mathbb{V}[\theta_d | \text{data}]^{\frac{1}{2}} \cdot y_d}{\hat{\mathbb{V}}[y_d]^{\frac{1}{2}} \cdot \mathbb{E}[\theta_d | \text{data}]}  \quad \forall d,
\end{align*}
using predictors and their posterior variances for any HB model. This constitutes a frequently used measure for small area model evaluation. However, the comparison between CV of model-based and design-based estimators might sometimes be spurious since the former could be design-biased even when the model is correctly specified. A discussion about it can be found in \citet{ferrante2017small}. Thus, we perform such comparison aware that the bias properties of our estimators have to be further deepened through the simulation study of Section \ref{designbasedsimulation}.

Diagnostics for each model, \texttt{looic} and \texttt{cvr}, averaged among all domains (\texttt{acvr}), are displayed in Table \ref{looicvr}, while the full distribution of \texttt{cvr} is displayed in Figure \ref{cvred}. Results show better goodness-of-fit and coefficient of variation reduction for the Flexible Beta model compared to the Beta one. This holds for all measures with the exception of the Gini index, where diagnostics do not vary significantly between models. Following the distributional analysis conducted in Section S.2 of the Supplementary Material, the sampling distribution of Gini index estimator does not show departures from Gaussianity in small samples, in comparison with other measures, presenting only light skewness and lepto-kurtosis. As a consequence, the employment of a mixture model seems to be irrelevant for a substantial improvement of the estimates.

The FB model allocates greater density on the right-hand tail of estimator distribution, being able to better capture it. This aspect is clear from density plots in Figure \ref{densities}, displaying direct estimates versus model-based estimates of Beta and FB models in the 107 domains. Notice that any data point refers to a domain, being the expected values of different posterior distributions as in (\ref{predictor}), thus let us consider them as global distributions not related to domain-specific posteriors. The FB model-based estimates tend to be greater than Beta model ones as clear from scatterplots in Figure \ref{densities}, capturing the right-hand tail of the distribution and avoiding to underestimate inequality, as it will be more intelligible in Section \ref{designbasedsimulation}.
The Gini index case shows, again, large similarities across the two models.

 \begin{figure}[h]
\centering
    \includegraphics[width=0.95\linewidth]{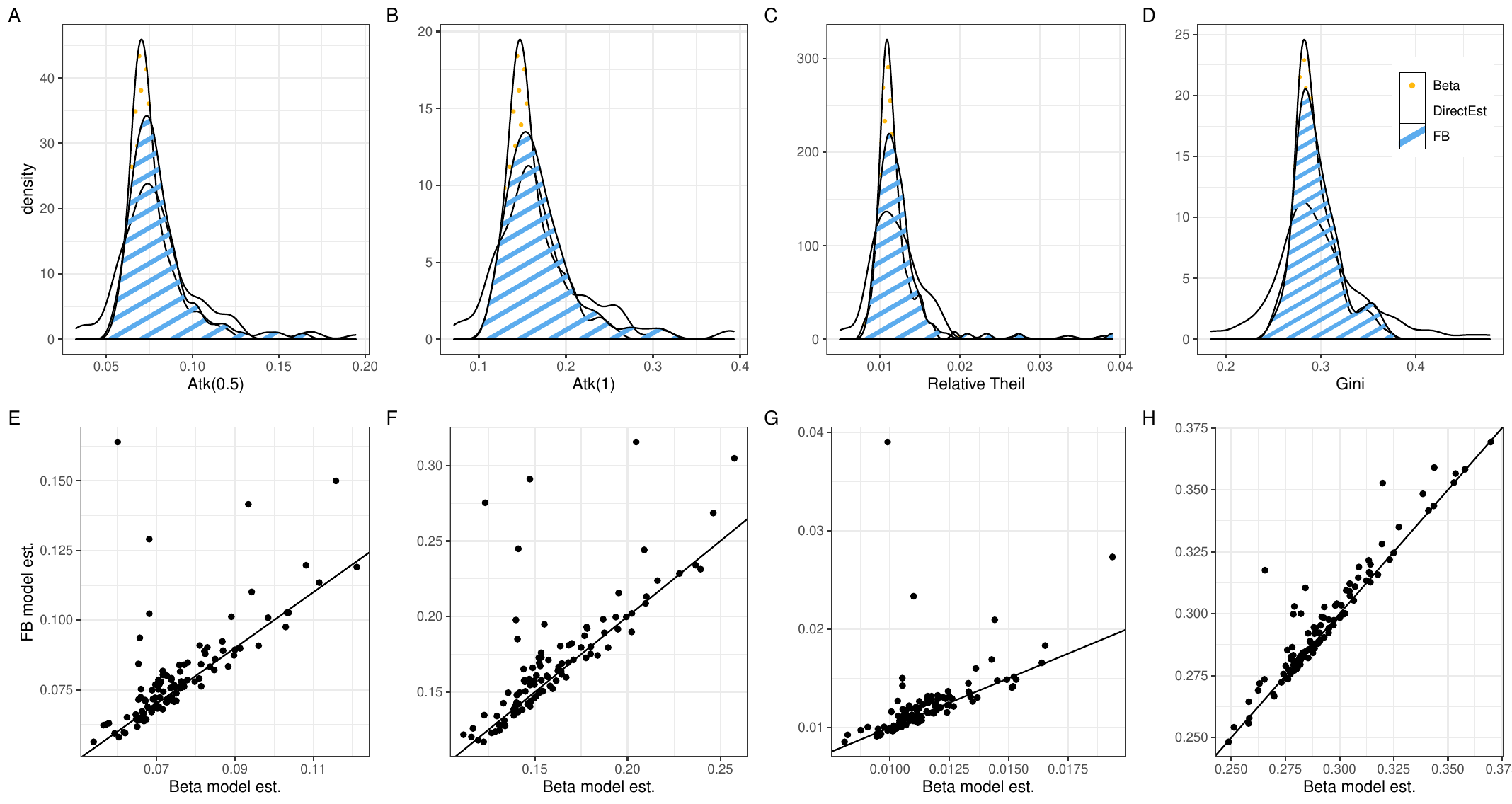}
  \caption{Densities of model-based estimates versus direct estimates and scatterplot of model-based estimates with bisector line}
  \label{densities}
  \end{figure}

\begin{figure}[h]
\centering
    \includegraphics[width=0.95\linewidth]{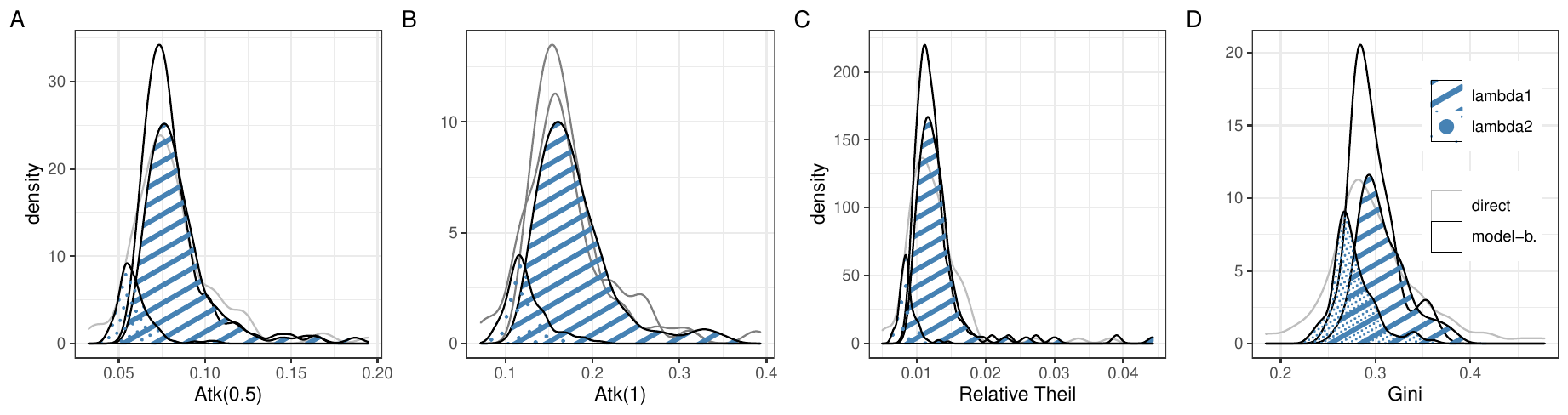}
 \caption{Posterior means distributions of the mixture components expected values $\lambda_{1d}$ and $\lambda_{2d}$, $\forall d$, weighted for $\mathbb{E}[p | \text{data}]$, in comparison with direct estimates and Flexible Beta model-based estimates}
  \label{mixt}
  \end{figure}

Deep diving on FB estimates, Figure \ref{mixt} displays posterior means distributions of mixture components expected values $\lambda_{1d}$ and $\lambda_{2d}$, $\forall d$, weighted for $\mathbb{E}[p | \text{data}]$, in comparison with direct estimates and Flexible Beta model-based estimates. The posterior means of the mixing coefficient are 0.83 for both Atkinson indexes, 0.85 for Relative Theil and 0.63 for Gini index, due to its more symmetric distribution. \textcolor{black}{We may notice from Figure \ref{mixt} that the mixture component embracing lower inequality values has always less weight.}

\begin{figure}
\centering
    \includegraphics[scale=0.6]{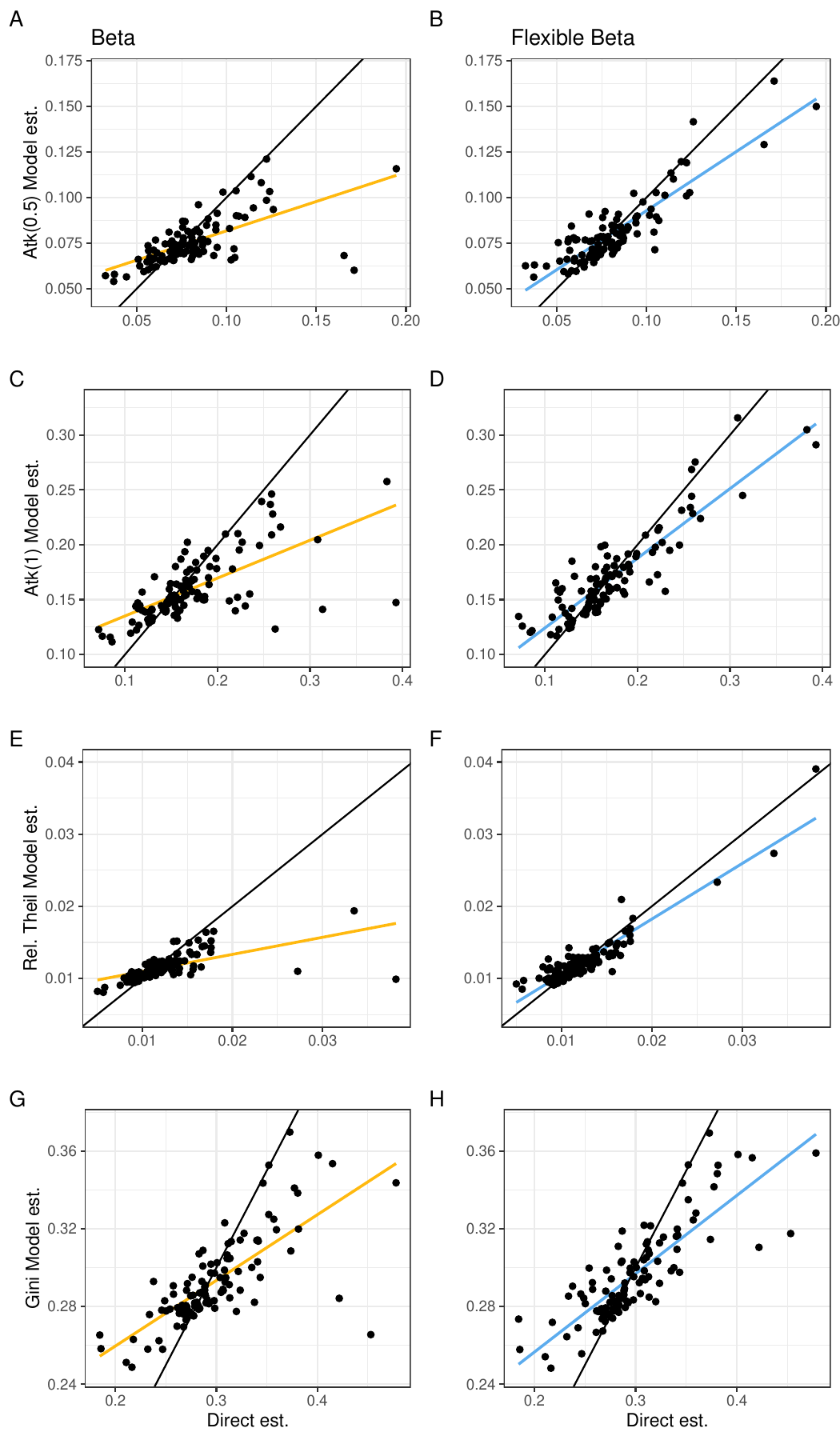}
  \caption{Direct estimates versus model-based estimates for each measure in Beta and Flexible Beta models: bisector in black, coloured linear regression line}
  \label{shrinking}
  \end{figure}

\begin{figure}
\centering
    \includegraphics[scale=0.6]{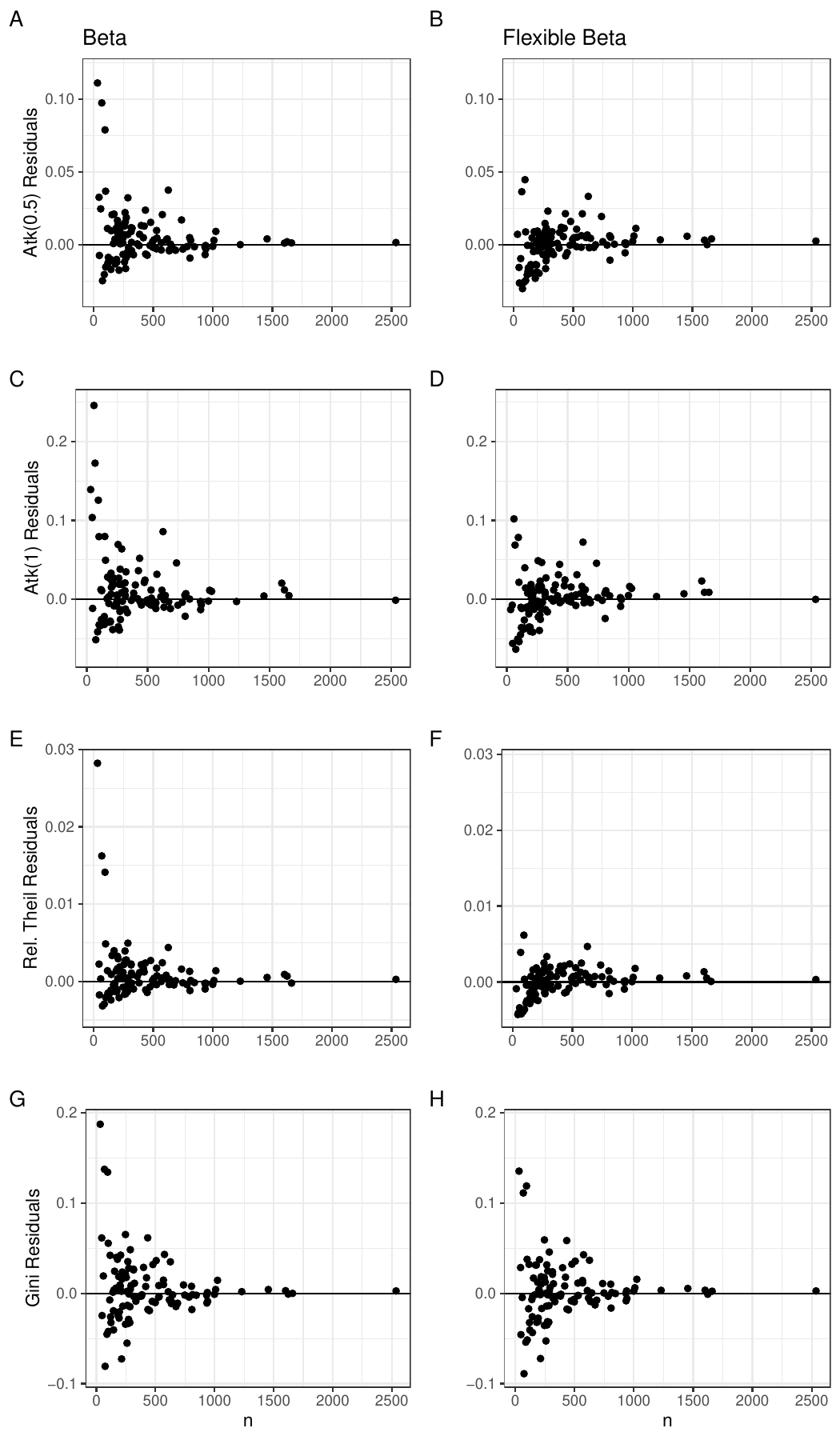}
  \caption{Design consistency check (sample sizes vs residuals) for each measure in Beta and Flexible Beta models}
  \label{designcon}

  \end{figure}

The shrinking process is displayed in Figure \ref{shrinking} for each model. It appears distinctly that estimates are more shrunk in the case of Beta model, as highlighted by the distance between linear regression and bisector lines. Moreover, a property that should be desirable for any small area model is the consistency among direct estimates and model-based ones, namely direct estimates outliers should not be completely pushed towards the opposite tail as model estimates. This consistency property does not hold in case of Beta models, having 2 or 3 top outliers pushed towards the lowest values of the distribution. This is due to the strong impact that auxiliary variables have on the outcome and to the little flexibility of the model. On the contrary, the FB model keeps its model estimates consistent with their input ones, operating overall less shrinkage. Another desirable property is the design consistency, i.e.  $(\hat{y}_d-\theta_d^{HB}) \rightarrow{} 0$ as long as $\tilde{n}_d$ increases. This property hold for all models, as clear from Figure \ref{designcon}. Notice that the magnitude of residuals for Beta models is relevant for domains with smaller sample sizes and strongly unbalanced on positive residuals. This makes sense, given the strong shrinkage operated on high outliers.
Posterior summaries of regression coefficients in the FB model are detailed and commented in Section S.4 of the Supplementary Material.

\section{Design-Based Simulation}
\label{designbasedsimulation}

A design-based simulation study has been carried out to evaluate the frequentist properties of the FB model-based estimators in comparison with the Beta ones. We consider the Italian EU-SILC sample as synthetic population and the 14 metropolitan cities and the remaining 21 administrative regions as synthetic domains.
In order to deal with a sufficient number of areas, assuring at the same time high variability of direct estimates (i.e. keeping synthetic population sufficiently large),
we pool 2009, 2013 and 2017 EU-SILC waves as independent and separate populations with a total of 105 domains.
The study is not based on generated data under some specific income distributional assumption, since the aim is to check whether this framework can work with close-to-reality income data, affected by peculiar problems, e.g. extreme values.

From each domain, $S=1,000$ samples have been repeatedly extracted by mimicking complex EU-SILC design, with stratification, multi-stage selection and the distinction between self-representative and non-self-representative strata.
We adopted three different scenarios for simulation, the first two involve different sampling rates, respectively 3\% and 5\%. Note that observations included in 3\% samples have been selected from those of 5\% samples to attenuate the effect of sampling variability. The third scenario consists of running the simulation on 5\% samples with smoothed income data, where an extreme values treatment (\texttt{evt}) has been performed as previously described in Subsection \ref{variancestimation} \citep{Finkelstein2006, masseran2019power}.
The drawing is done at the household level, with a total of extracted individuals in each domain ranging respectively from 9 to 115 (median equals 31) for 3\% and from 9 to 228 (median equals 63) for 5\%.

Bias-corrected inequality estimators have been calculated for any extracted sample, and a suitable set of covariates have been selected among the ones cited in Subsection \ref{auxiliaryvariables}. Covariates have been calculated at the corresponding geographical detail for the 105 synthetic domains.  At a lower level of geographical disaggregation, such as in this case (NUTS-2 regions), the correlation among covariates and between covariates and response is stronger in comparison with our application setting (NUTS-3 regions). This is coherent since raw estimates are measured at macro level, inducing less error. For any iteration $s=1, \dots, S$, Beta and FB model has been estimated with a fixed set of covariates. The sole distinction with the setting adopted in Section \ref{application} regards the prior of mixing coefficient $p$ in \eqref{priorp}, substituted with $p \sim \text{Beta}(2,2)$, in order to speed up convergence and save computational time.

Considering the generic model-based estimate at iteration $s$ for domain $d$ as $\hat{\theta}^{HB}_{ds}$ and the corresponding population value $\theta_d$, we define Relative Bias (RB), Absolute Relative Bias (ARB), Mean Squared Error (MSE), Relative Mean Squared Error (RMSE) and Average Effect (AEFF) as follows:
\begin{align*}
&\text{RB}(\hat{\theta}^{HB}_d)=\frac{1}{S} \sum_{s=1}^{S}  \bigg( \frac{\hat{\theta}_{ds}^{HB}}{\theta_d} -1 \bigg), \\
& \text{ARB}(\hat{\theta}^{HB}_d) =\bigg| \frac{1}{S} \sum_{s=1}^{S}  \bigg( \frac{\hat{\theta}_{ds}^{HB}}{\theta_d} -1 \bigg) \bigg|,
\end{align*}

\begin{align*}
&\text{MSE}(\hat{\theta}^{HB}_d) =  \frac{1}{S} \sum_{s=1}^{S} ( \hat{\theta}_{ds}^{HB}-\theta_d)^2  , \\
&\text{RMSE}(\hat{\theta}^{HB}_d) = \frac{\text{MSE}(\hat{\theta}^{HB}_d)}{\theta_d^2}, \\
&\text{AEFF}(\hat{\theta}^{HB}) =\sqrt{\frac{\sum_{d=1}^{D} \text{MSE}(y_d) }{\sum_{d=1}^{D} \text{MSE}(\hat{\theta}^{HB}_d)}}.
\end{align*}

\noindent Lastly, we consider the frequentist coverage of credible intervals defined by the $\alpha/2$ and $1-\alpha/2$ quantiles of the posterior of $\theta_d$,
\begin{align*}
\text{Coverage}_{1-\alpha} (\hat{\theta}^{HB}_d) &=  \frac{1}{S} \sum_{s=1}^{S}   \mathds{1} (\theta_d \in [ Q_{\alpha/2}[\theta_{ds} | \text{data} ],  Q_{1-\alpha/2}[\theta_{ds} | \text{data} ] ]),
\end{align*}
where $Q_{\pi}[\theta_{ds}| \text{data}]$ denotes the posterior quantile of order $\pi$ of $\theta_{ds}$.
The nominal coverage probability $1-\alpha$ is chosen to be equal to 0.95.

Simulations results are fully described for any setting in Table \ref{simresults}. RB, ARB, RMSE and Coverage are reported on average over the 105 simulations domains, showing that FB estimators outperform the Beta ones.

\begin{table}[h]
\footnotesize
\centering
\begin{tabular}{llllllll}
Measure & Scenario & & $\bar{\texttt{ARB}}$\%& $\bar{\texttt{RB}}$\%& $\bar{\texttt{RMSE}}$\%  & \texttt{AEFF} & 
$\bar{\texttt{Cov.}}$ 95\% \\
  \hline
      Atk(1) & \texttt{evt} &&&\\
      &  &direct&   1.44& -0.99 & 13.99     &&\\
  &     &Beta & 9.90 & -3.06 & 2.10 & 2.21 
       & 87.00 \\
&&  FB & 8.78 & -0.22 & 1.81 & 2.65 
& 89.70  \\
&  5\%&&&\\
 &    &direct&   2.00 &-1.42& 19.48      &&\\
 &  &Beta  & 11.15 & -2.46 & 2.44 & 2.46 
 &   84.98  \\
&&  FB & 9.38 & -0.12 & 2.04 & 2.97 
& 90.28  \\
& 3\%&&&\\
&   &direct& 3.89 & -3.68& 32.77\\
& &Beta & 11.49 & -2.17 & 2.61 & 3.07  
&  85.82  \\
&&  FB & 9.12 & -0.74 & 2.33 & 3.54  
& 91.50  \\
    \hline
        Atk(0.5) & \texttt{evt}&&&\\
 &        &direct&   1.69 & -1.24& 15.87     &&\\
 &        &Beta  & 9.56 & -2.75 & 1.96 & 2.50 
       & 87.64  \\
 & &FB & 8.41 & -0.23 & 1.67 & 2.95 
  & 91.19  \\
   & 5\%&&&\\
     &  &direct&    2.40 & -1.92&  21.90      &&\\
    & &Beta & 10.85 & -2.41 & 2.35 & 2.60
    &  85.25  \\
 & &FB & 8.85 & -0.04 & 1.85 & 3.34
  & 92.07  \\
  & 3\%&&&\\
  & &direct& 4.78 & -4.59& 34.72&&      \\
 & &Beta & 11.06 & -2.42 & 2.47 & 3.14 
 & 86.30  \\
 & &FB & 8.67 & -1.28 & 2.08 & 3.79 
& 93.01 \\
 \hline
     Relative Theil  & \texttt{evt}&&\\
     &     &direct&   9.69 &-5.14& 18.73    &&\\
 &       &Beta& 12.55 & -6.52 & 3.59 & 1.94 
      & 70.72  \\
 &   &FB & 9.71 & -3.72 & 3.31 & 2.22 
  & 73.31 \\
  &  5\%&&\\
     &    &direct&     3.63 &-3.23&  27.87     &&\\
       &  & Beta  & 12.30 & -5.02 & 3.21 & 2.36 
       & 81.79 \\
  & & FB  & 7.69 & -0.80 & 1.90 & 3.75 
 & 92.84  \\
   & 3\%&&&\\
  &  &direct&  7.48 &-7.35& 39.90&&        \\
   &  &  Beta & 12.71 & -4.82 & 3.38 & 2.62 
     & 80.76  \\
    & & FB & 8.32 & -3.33 & 2.55 & 3.44 
   & 90.68  \\
  \hline
  Gini  &  \texttt{evt}&&\\
  &  &direct&2.58 &-1.13& 4.75&&     \\
   &  &  Beta & 4.60 & -1.06 & 0.53 & 2.85 
   & 91.50  \\
 &    & FB & 4.43 & -0.28 & 0.50 & 3.01 
   & 93.23  \\
 &  5\%&&\\
  &  &direct&     3.29 &-0.10& 7.43    &&\\
 &  &Beta  & 5.10 & -0.58 & 0.69 & 3.16 
 & 91.23  \\
 & &FB & 4.90 & 0.68 & 0.67 & 3.33 
& 93.04  \\
  & 3\%&&\\
 &  &direct& 6.20 &-5.99& 9.02 &&\\
 & &Beta & 6.21 & -5.19 & 0.99 & 2.75  
& 89.84  \\
 & &FB & 5.79 & -4.80 & 0.92 & 2.89  
& 92.09  \\
 \hline
\end{tabular}
\caption{Absolute Relative Bias, Relative Bias, Relative Mean Squared Error, Average Effect, jointly with the 95\% coverage on average of direct estimators and model-based estimators related to both Beta and FB models for extreme value treated data (\texttt{evt}), 5\% and 3\% sampling rates. Quantities denoted with the bar are averaged among all areas.}
\label{simresults}
\end{table}

\begin{figure}
\centering
    \includegraphics[scale=0.7, angle=90
    ]{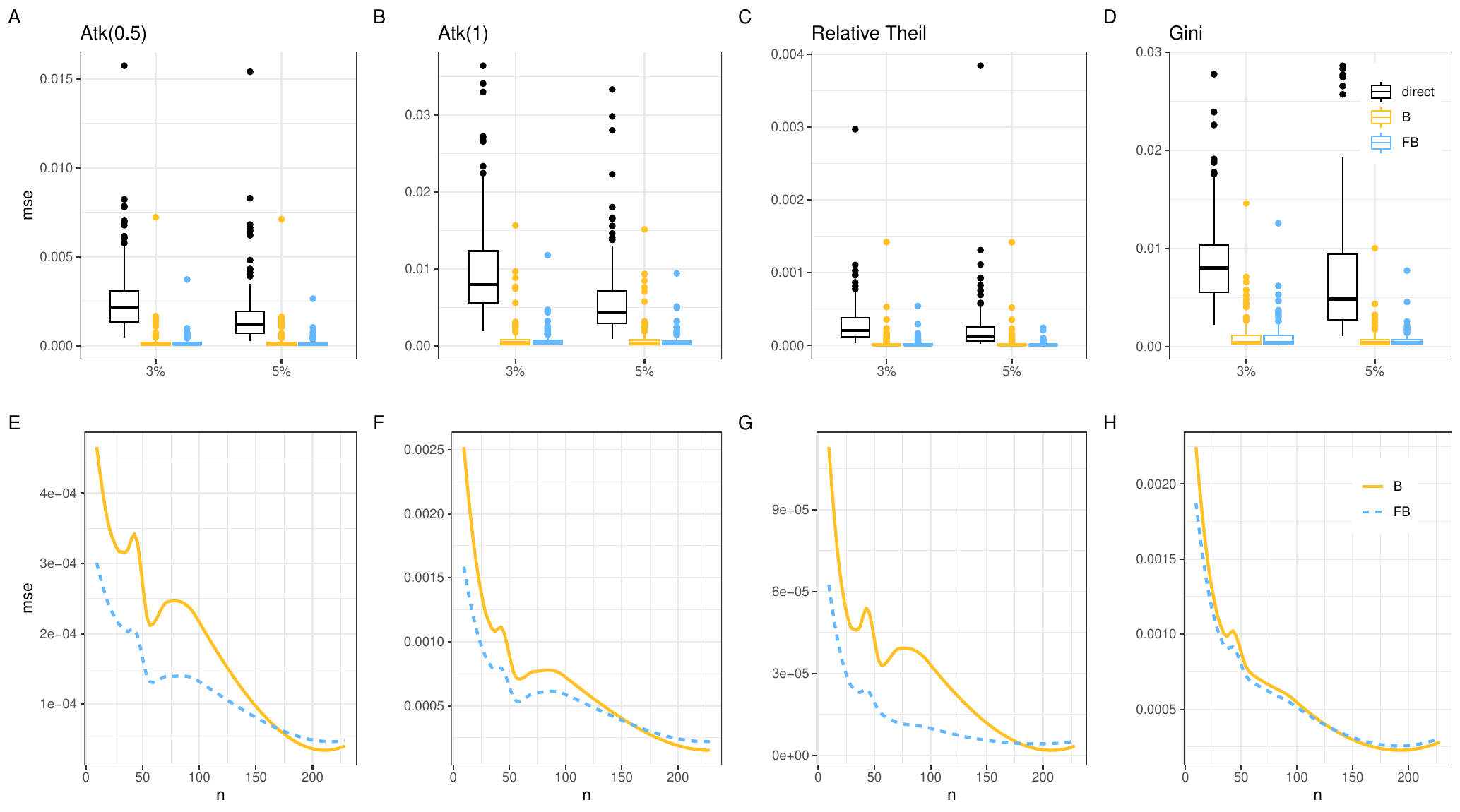} 
\small  \caption{MSE for each area. Plots on the top row show direct estimators values versus model-based estimators ones, while bottom row plots show model-based estimates MSE versus the sample sizes, depicted through smoothed lines}
  \label{mse}
  \end{figure}

\normalsize
Focusing on estimates reliability, both Beta and Flexible Beta models perform significantly better
than direct estimators: RMSE and AEFF show a great error reduction for all measures. Among them, the FB estimators perform better than the Beta ones in all cases.
Moreover, the FB estimators show a noticeable bias decrease compared to the Beta one, as clear from ARB and RB values in Table \ref{simresults}, concerning both magnitude and direction. This confirms the clues of inequality underestimation under a Beta model, notwithstanding the measure adopted, and shows that the FB model consistently reduces this underestimation. The bias improvement is at the expense of a slight variance increase, but the bias-variance trade-off favors the FB model, as notable from RMSE and AEFF.

The full distribution of MSEs of direct and model-based estimators related to the different domains is depicted by boxplots in Figure \ref{mse}. Firstly, all distributions show heavy right-tails, with several outlier domains having great error levels. Again, the error reduction induced by both small area models is noticeable, allowing estimators to borrow strengths across areas.
Specifically, while RMSEs, displayed in Table \ref{simresults}, indicate on average a moderate error improvement for the FB model, the full distributions show a great reduction in case of outlier domains. This reduction, as clear from the bottom row plots of Figure \ref{mse}, takes place in case of domains with the smallest sized samples. The greatest MSE reduction regards the Relative Theil Index, the lowest the Gini index.

  \begin{figure}[h]
\centering
    \includegraphics[width=0.95\linewidth]{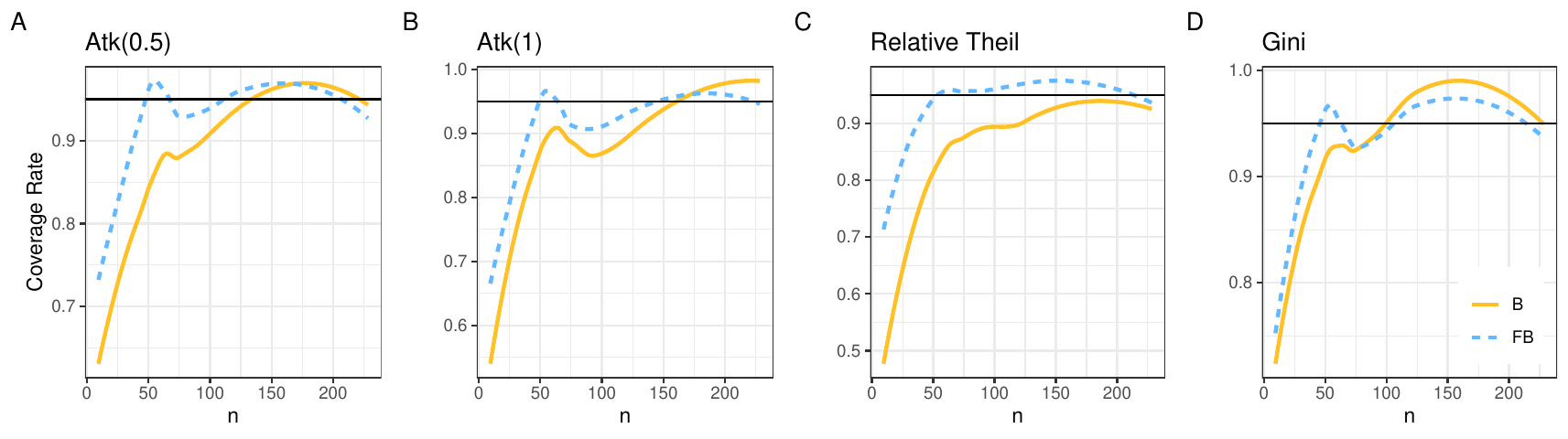}
  \caption{Coverage rate versus sample sizes for each measure and each model in 5\% simulation samples depicted through smoothed lines, black line fixed at the nominal level 0.95}
  \label{coverage}

  \end{figure}

 Flexible Beta models produce credible intervals that exhibit a noticeable better performance in terms of coverage, in some cases outperforming Beta intervals coverage by more than 10\% on average. Its trend over the sample sizes is displayed in Figure \ref{coverage}. While FB coverage rates converge to their nominal level in correspondence to 50 individual-sized samples, the Beta ones converge near 100 sized samples in case of Gini Index, near 130/150 in case of Atkinsons and Relative Theil Indexes.

 The similarity among Beta and FB model-based estimates for the Gini index is confirmed also in the simulation setting. Nevertheless, the FB model has higher coverage in case of small samples. Concerning different  simulation settings, the Relative Theil direct estimators show noticeably high bias in case of extreme value treated setting in comparison with other settings. This could be due to a failure of preliminary bias correction on direct estimators; indeed, input estimators does not satisfy unbiasedness, not allowing for comparability. The extremely low coverage of both model-based estimators is indeed justified by the high bias. Generally speaking, the performance gap between Beta and FB models increases at decreasing sampling rates/sizes and with no smoothed data. This denotes a progressive failure of the Beta model under the mentioned conditions, as clear from how fast its diagnostics get worse at varying settings.

\section{Conclusions}
\label{conclusions}

The reduction of inequalities, both within and between countries, is a prerequisite for achieving the
Sustainable Development Goals of the 2030 Agenda for Sustainable Development, adopted by all
United Nations Member States. At regional level, an increase in disparities within regions has been observed, whereas regional disparities between European countries are
gradually decreasing. Several low-growth regions exist within EU member states and,
among the richest states, it is possible to find areas characterised by
high levels of poverty and inequality (often post-industrial or rural ones). In this context, inequality indicators at a regional-level
breakdown would allow us to shift towards a more comprehensive and multifaceted view of territorial
convergence in the EU and to better understand causal mechanisms, essential for regional-targeted policies.

In this study, we propose a SAE model that aims at obtaining reliable estimates of the
most common inequality indicators: 
the Gini index, the Relative Theil
index and two Atkinson indexes defined for two different values of the inequality aversion
parameter. By considering that inequality estimators  
are unit-interval defined, skewed and heavy-tailed, we propose a
FB small area model.
The results are really encouraging as the estimates we obtain
outperform in different ways the most common Beta small area model, generally used for parameters defined on the unit interval.
We may define the FB model as a good general-purpose model for unit interval data as the two-component mixture is able to capture skewness and heavy tails in a satisfying way, guaranteeing at the same time analytic
tractability.

Our findings provide a basis for further research focused on multiple directions. An inequality mapping based on obtained estimates could be used to single out at first a territorial disparity analysis within the country, complementing the usual consideration of disparity between countries.
Furthermore, the set of inequality measures should be properly complemented by
quantile-based inequality indexes that, by focusing on distribution tails, are able to capture different
aspects of the income distribution with respect to concentration indexes. Quantile-based
indexes are not defined on the unit interval support and thus their estimation has to rely on different likelihood assumptions.  
Lastly, given the longitudinal nature of the EU-SILC survey, a model extension in this sense naturally follows, by considering subsequent waves at the same time.

\section*{Acknowledgements}

The work of Silvia De Nicolò was partially funded by the ALMA IDEA 2022 grant (title: "Social exclusion and territorial disparities: poverty and inequality mapping through advanced methods of small area estimation", project J45F21002000001), part of the European Union - NextGenerationEU funding.

\bibliographystyle{plainnat}
\bibliography{library.bib}

\end{document}